\documentclass[journal]{IEEEtran}

\usepackage{tikz}
\usetikzlibrary{positioning}

\usepackage{setspace}
\usepackage{color}
\usepackage{graphicx}
\usepackage{url}
\usepackage{amssymb}
\usepackage{amsmath}
\usepackage{bbm}
\usepackage{amsthm}
\usepackage{subfigure}
\usepackage{comment}
\usepackage{amsfonts,amssymb,epsfig,epstopdf}
\usepackage{mathrsfs}
\usepackage{epsf,psfrag}
\usepackage{array}  
\usepackage{algorithm}
\usepackage{algorithmic}
\usepackage{tabularx}
\usepackage{listings}
\usepackage{makecell}
\usepackage{multirow}
\usepackage[switch]{lineno}
\usepackage{float}
\usepackage{hhline}
\usepackage[normalem]{ulem}

\newtheorem{theorem}{Theorem}
\newtheorem{lemma}[theorem]{Lemma}

\newcommand{\squishlist}{
	\begin{list}{$\bullet$}
		{ \setlength{\itemsep}{2.5pt}      \setlength{\parsep}{3pt}
		  \setlength{\topsep}{3pt}       \setlength{\partopsep}{0pt}
		  \setlength{\leftmargin}{1.0em} \setlength{\labelwidth}{1em}
		  \setlength{\labelsep}{0.5em} } }

\newcommand{\squishlisttwo}{
	\begin{list}{$\bullet$}
		{ \setlength{\itemsep}{2.25pt}    \setlength{\parsep}{0pt}
		  \setlength{\topsep}{5pt}     \setlength{\partopsep}{0pt}
		  \setlength{\leftmargin}{2em} \setlength{\labelwidth}{1em}
		  \setlength{\labelsep}{0.5em} } }
	
\newcommand{\squishend}{
	\end{list}  }
	
\makeatletter
\def\blfootnote{\xdef\@thefnmark{}\@footnotetext}
\makeatother
	
\begin{document}

%
\title{\textit{MLSNet:} A Policy Complying Multilevel Security Framework for Software Defined Networking}

%


\author{Stefan Achleitner, Quinn~Burke,~\IEEEmembership{Student Member,~IEEE,}\\ Patrick~McDaniel,~\IEEEmembership{Fellow,~IEEE,} Trent~Jaeger,~\IEEEmembership{Member,~IEEE,}\\ Thomas~La~Porta,~\IEEEmembership{Fellow,~IEEE,} and Srikanth~Krishnamurthy,~\IEEEmembership{Fellow,~IEEE}\\
}



%
%
\maketitle
\blfootnote{
This research was sponsored by the U.S. Army Combat Capabilities Development Command Army Research Laboratory and was accomplished under Cooperative Agreement Number W911NF-13-2-0045 (ARL Cyber Security CRA). The views and conclusions contained in this document are those of the authors and should not be interpreted as representing the official policies, either expressed or implied, of the Combat Capabilities Development Command Army Research Laboratory or the U.S. Government. The U.S. Government is authorized to reproduce and distribute reprints for Government purposes notwithstanding any copyright notation here on.

Stefan Achleitner was with the Department of Computer Science and Engineering, The Pennsylvania State University, University Park, PA 16802 USA. He is now with Palo Alto Networks, Inc., Santa Clara, CA 95054 USA (e-mail: stefan@stefanachleitner.com).

Quinn Burke, Patrick McDaniel, Trent Jaeger, and Thomas La Porta are with the Department of Computer Science and Engineering, The Pennsylvania State University, University Park, PA 16802 USA (e-mail: qkb5007@psu.edu; mcdaniel@cse.psu.edu; trj1@psu.edu; tfl12@psu.edu).

Srikanth Krishnamurthy is with the Department of Computer Science and Engineering, University of California, Riverside, Riverside, CA 92521 USA (e-mail: krish@cs.ucr.edu).
}

\begin{abstract}
Ensuring that information flowing through a network is secure from manipulation and eavesdropping by unauthorized parties is an important task for network administrators. Many cyber attacks rely on a lack of network-level information flow controls to successfully compromise a victim network. Once an adversary exploits an initial entry point, they can eavesdrop and move laterally within the network (e.g., scan and penetrate internal nodes) to further their malicious goals. In this paper, we propose a novel multilevel security (MLS) framework to enforce a secure inter-node information flow policy within the network and therein vastly reduce the attack surface available to an adversary who has penetrated it.
In contrast to prior work on multilevel security in computer networks which relied on enforcing the policy at network endpoints, we leverage the centralization of software-defined networks (SDNs) by moving the task to the controller and providing this service transparently to all nodes in the network.
Our framework, \textit{MLSNet}, formalizes the generation of a policy compliant network configuration (i.e., set of flow rules on the SDN switches) as network optimization problems, with the objectives of (1) maximizing the number of flows satisfying all security constraints and (2) minimizing the security cost of routing any remaining flows to guarantee availability. We demonstrate that MLSNet can securely route flows that satisfy the security constraints (e.g., $>80\%$ of flows in a performed benchmark) and route the remaining flows with a minimal security cost.
\end{abstract}


\section{Introduction}
\label{sec:intro}
Ensuring that information flowing through a network is secure from manipulation and eavesdropping by unauthorized parties is an important task for network administrators. Many attacks against modern networks rely on a lack of network-level information flow controls to infiltrate an organizational network. Here, adversaries initially subvert edge defenses to target and compromise an internal node. Once inside the network, the adversary can probe network nodes or eavesdrop on flows to penetrate further into the network~\cite{jang2014survey}. This adversary-enabling freedom of movement and lack of secure routing (to prevent eavesdropping) can be cast as a classical {\it information flow} problem in security~\cite{denning1976lattice}.

Even with defenses such as firewalls, information flow control in networks often fails: configuration is error-prone~\cite{yuan2006fireman}, and compromised internal hosts may initiate flows that never have to cross a firewall boundary~\cite{spitzner2003honeypots}. Thus, adversaries can exploit firewall rule conflicts to exfiltrate information, and internal adversaries can eavesdrop and move laterally (i.e., scan and penetrate internal nodes) within their network boundary without restriction. Fundamentally, they are enabled by a lack of security policy governing what flows are permitted and what paths they may take in the network.

Multilevel security (MLS) provides the means to enforce such a policy. A multilevel security framework controls information flow among entities of different security classes with security labels (i.e., levels and categories) assigned to those entities. In fact, multilevel security already plays a critical role in controlling access to information for both military personnel and employees of commercial businesses with different levels of clearance~\cite{saydjari2004multilevel}. Common use cases include controlling file access in an operating system~\cite{loscocco2001security} or table access in a relational database~\cite{qian1997semantic}. 
The notion of multilevel security can also be applied to computer networks, where the MLS policy dictates which nodes are allowed to communicate, what type of traffic they may exchange, and what paths the flows may take in the network. This strategy can prevent eavesdropping and unrestricted lateral movement that plague modern networks.

\textit{Lu et al.}~\cite{lu1990model} envisioned such a model that enforces the information flow policy at network endpoints; however, the scale and dynamic behavior of modern networks make deploying such an enforcement mechanism on every endpoint impractical. Despite this, the inherent centralization of software-defined networks (SDNs) allows enforcement of a network-level MLS policy in a scalable and efficient manner. Determination of whether or not flows are permitted can be done by the controller, and the policy can be enforced by the data-plane switches in the form of flow rules---which allows the service to be provided transparently to the entire network.

Thus, in this paper, we propose an SDN-based MLS framework to enforce an inter-node information flow policy that preserves confidentiality. The challenge here is to fit the organizational needs by allowing entities to exchange permitted flows while also configuring the network (by leveraging flow rules) to be policy compliant. Permitted flows between two endpoints may not always find a secure path due to limited network resources (e.g., link capacity). Then, to guarantee availability, a flow may have to be routed through an insecure path. We refer to such a situation as a \textit{policy conflict}, and each conflict imposes a \textit{security cost} in terms of the risk the flow is being exposed to.

Unlike prior work~\cite{lu1990model}, we approach the challenge of securing information flow in the network by considering two optimization models: one that can provide a secure network configuration (i.e., composition of flows rules) that obeys the security policy, supplemented by a model that can minimally relax the security policy to ensure that every flow can be routed. The key contributions are:

\squishlisttwo
{\color{black}
\item{An optimization model to maximize the number of flows routed according to the given security policy in an SDN.}
\item{An optimization model to minimize the security cost of routing any remaining flows to guarantee availability.}
\item{A method for constructing flow rules which adhere to a given security policy.}
\item{{A comprehensive evaluation of MLSNet's ability to generate policy compliant network configurations and resolve policy conflicts in realistic network topologies.}}
}
\squishend

\begin{table}[t]
	{{
		\centering
		\caption{Nomenclature and notation.}
		\label{tab:notation}
		\begin{tabular}{|l|l|}
			\hline
			\bfseries Notation & \bfseries Description \\
			\hline
			$V$ & Set of vertices in network graph $G=\{V,E\}$\\
			\hline
			$E$ & Set of edges in network graph $G=\{V,E\}$\\
			\hline
			$F$ & Set of packet flows to be accommodated\\
			\hline
			$R$ & Set of matching fields in a flow rule\\
			\hline
			$A$ & Set of action fields in a flow rule\\
			\hline
			$S$ & Set of subjects\\
			\hline
			$O$ & Set of objects\\
			\hline
			$C$ & Set of security categories\\
			\hline $d^f$ & Size of flow $f \in F$\\
			\hline $\kappa_{ij}$ ($\widetilde{\kappa}_{ij}$) & Residual capacity of link $(i,j)$, $(i,j) \in E$\\
			\hline $\sigma_{i}$ & Security level of node $i$, $i \in V$\\
			\hline $\lambda^c_{i}$ & Security category $c$ at node $i$\\
			\hline $L$ & Set of labels that form the lattice\\
			\hline
		\end{tabular}
	}}
\end{table}
\section{Definitions and Background}
\label{sec:background}
In this section, we extend prior work's~\cite{lu1990model} terms and notations (Table~\ref{tab:notation}) to an SDN setting and provide background on related network security threats and defenses, and MLS.

\subsection{Term Definitions}
\label{sec:definitions}
{\color{black}
\noindent \textbf{Node.} A resource connected to a network (e.g., a user, server, router, or SDN switch).\newline\noindent
\textbf{Subject.} A node that initiates communication to other nodes in the network.
\newline\noindent
\textbf{Object.} A node that either provides (\textit{provider}) and/or receives (\textit{receiver}) information to/from subjects\newline\noindent
\textbf{Forwarding Node.} A node (SDN switch) that processes incoming flows according to the installed flow rules.\newline\noindent
\textbf{Controller.} An application in the SDN control plane that has a global view of the topology and installs flow rules to forwarding nodes based on the security policy.\newline\noindent
\textbf{Security Levels.} Hierarchical attributes (e.g., top-secret, public) that indicate relative authorization power.\newline\noindent
\textbf{Security Categories.} Non-hierarchical attributes (e.g., TCP, IP) that offer finer-grained authorization besides the security level. In MLSNet, security categories are only assigned to objects and subjects but not forwarding nodes. \newline\noindent
 \textbf{Security Label.} The security level and categories combined, used by the controller to admit or deny flows.

{\color{black}
\subsection{Network Threats to Confidentiality}
\label{sec:securitydiscussion}
{\color{black}
Confidentiality ensures that information is only being accessed by authorized parties. In the context of networking, preserving confidentiality means that only explicitly allowed communication can flow between any two nodes in the network to prevent data from falling into the hands of untrusted entities. A lack of a formal communication policy to realize this may allow adversaries who have compromised internal nodes to explore the network or eavesdrop on flows. For example, in Figure~\ref{fig:mls_scenario_1}, a compromised trusted node in the software center may be able to probe other nodes in the software and commercial centers as they are all behind the same firewall boundary; the firewall itself cannot prevent the adversary from probing all nodes on TCP port 22. Indeed, this is possible regardless of the (implicit) security level of the nodes; however, we can reduce attacker capabilities (enforce least-privilege) with multilevel security.

\begin{figure}[t]
	\centering
	\includegraphics[width=0.42\textwidth]{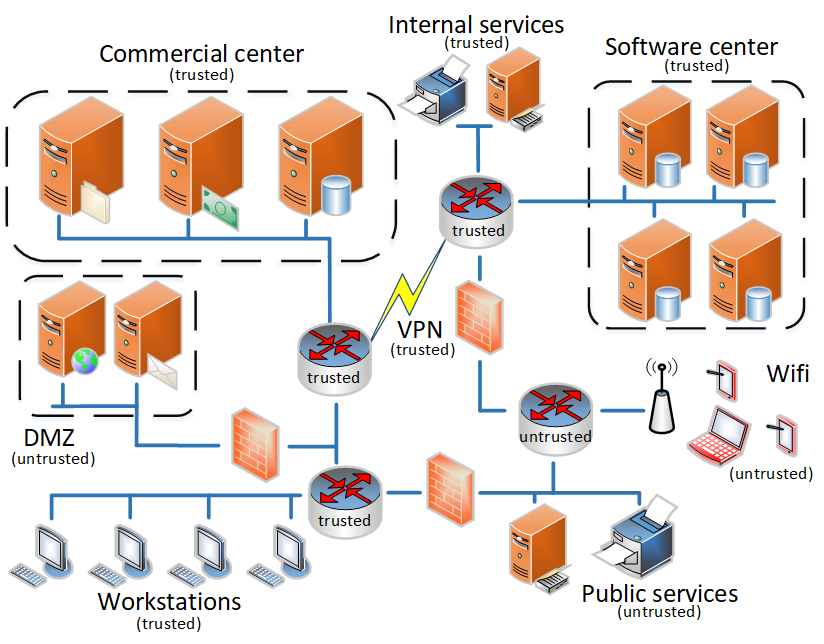}
	\caption{Corporate network scenario.}
	\label{fig:mls_scenario_1}
\end{figure}

Recent work has demonstrated the ability of an adversary to freely probe within their network boundary to recover sensitive information about the network~\cite{hong2015poisoning,yoon2017flow,dhawan2015sphinx,porras2012security,hu2014flowguard}, including active network hosts and even switch flow table rules~\cite{achleitner2017adversarial}.
We observe that although discovered attacks on networks in the literature pursue different goals, the strategies of those posing a threat to confidentiality can be reduced to a small set of techniques.

\textbf{Packet Spoofing.} Spoofing packets is the most common technique. By spoofing, adversaries may be able to impersonate other nodes to escalate privilege~\cite{dawoud2010infrastructure} or leak information to untrustworthy nodes or outside of the network~\cite{tankard2011advanced}.

\textbf{Lateral Movement.} Adversaries can also move laterally by probing many other nodes in the network. This nonessential communication may allow them to extract sensitive information from nodes of higher security levels or compromise nodes and escalate privilege to move deeper into the network~\cite{achleitner2017adversarial}.

\textbf{Man-in-the-Middle.} Adversaries can also position themselves as a man-in-the-middle (MiTM), silently eavesdropping on communications traversing them or within their broadcast domain~\cite{antikainen2014spook,khan2017topology}.
}

\subsection{Proposed Defenses}
\label{sec:proposeddefenses}
{\color{black}Although there have been defenses proposed against some of the discussed attack techniques, they are limited in their ability to preserve confidentiality.

\textbf{Source validation.} To address the issue of packet spoofing, source validation seeks to only permit packets into a network which's source IP is valid on the given network interface. This is typically implemented via ingress filtering~\cite{senie1998network} in wide-area networks; however, it is difficult to implement, especially in data-center networks~\cite{savu2011cloud}, and does not prevent adversaries from spoofing nodes within their own subnetwork.

\textbf{Firewalling.}
The primary purpose of a firewall is to prevent unauthorized packets from entering a network or subnetwork. However, firewalling is limited with respect to preventing lateral movement as configuration is error-prone~\cite{yuan2006fireman}, and compromised internal hosts can still probe within their network boundary~\cite{spitzner2003honeypots} to compromise internal nodes.

\textbf{Encryption.}
Active man-in-the-middle attacks (i.e., those staged by spoofing) may be mitigated with source validation; however, preventing passive MiTM (eavesdroppers) is difficult. Even with services such as encryption, adversaries can still perform traffic analysis to extract sensitive information~\cite{feghhi2016web}.
}

\begin{table}[t]
	\centering
	\caption{Defenses against discussed attack techniques.}
	\small
	\label{tab:sdn_defense}
	\begin{tabular}{| c | c | m{0.5cm} | m{0.5cm} | m{0.5cm} |}
		\cline{2-5}
		\multicolumn{1}{c|}{} & \textbf{MLSNet} & \textbf{\cite{senie1998network}} & \textbf{\cite{spitzner2003honeypots}} & \textbf{\cite{feghhi2016web}} \\ \hline
		\textbf{Packet spoofing} & $\times$ & $\times$ & - & - \\ \hline
		\textbf{Lateral movement} & $\times$ & - & $\times$ & - \\ \hline
		\textbf{Man-in-the-middle} & $\times$ & - & - & $\times$ \\ \hline
	\end{tabular}
\end{table}
{\color{black}
\subsection{Preserving Confidentiality with Multilevel Security}
Broadly speaking, existing defenses solve distinct problems and only partially address the issue of confidentiality. Adversaries are enabled by a lack of policy preventing them from probing network nodes and eavesdropping on communications. What is needed are means to specify what flows are permitted and what paths they may take in the network.
\newline \indent
\textbf{MLS.} A multi-level security policy provides the means to prevent these problems with a secure flow model between entities that are assigned specific security labels (i.e., a level and categories). The security labels form a lattice structure, which reflects a hierarchical ordering of their relative authorization power. We consider a node's label to be higher than another node's if the former's level is greater than or equal to, and the categories form a superset of, the latter's. With respect to confidentiality, information should only flow to nodes with the same or higher security label to prevent the potential leakage of sensitive data to nodes of lower security labels. This is typically summarized as \textit{"no read up, no write down"}.
\newline \indent
\textbf{Network MLS.} Multilevel security already plays a critical role in controlling access to files and databases in military and commercial business contexts~\cite{saydjari2004multilevel,loscocco2001security,qian1997semantic}. This notion can also be applied to computer networks to prevent the eavesdropping and unrestricted lateral movement that plague modern networks. For example, nodes with lower security levels should not be able to probe or communicate with nodes of higher levels on specific TCP ports, and sensitive (e.g., top-secret) flows should not traverse a node of lower security level. In this context, for communication to be permitted and routed between two nodes, both nodes and any intermediate nodes must adhere to the \textit{"no read up, no write down"} policy.
\newline \indent
\textit{Lu et al.}~\cite{lu1990model} envisioned such an MLS model that enforces the information flow policy at network endpoints. The problem with this approach is that the scale and dynamic behavior of modern networks make deploying such an enforcement mechanism on every endpoint impractical. However, the inherent centralization of software-defined networks (SDNs) allows enforcement of a network-level MLS policy in a scalable and efficient manner. Determination of whether or not flows are permitted can be done by the controller, and the policy can be enforced by the data-plane switches in the form of flow rules. This offers the significant advantage over previous work of allowing the service to be provided transparently to the entire network, because network devices do not have to run specialized software. The controller's global view of the network also offers greater flexibility as changes to labels and policy can be reflected by simple changes to flow rules, as opposed to manually re-configuring individual devices.
}

Ultimately, multilevel security can ensure that the network achieves (to the degree possible) least-privilege isolation, where only explicitly allowed communication can flow within the network and must flow through secure paths. Hence, it provides for maximal isolation from unauthorized parties and therefore the smallest possible threat surface. Here, we can mitigate lateral movement (by enforcing least-privilege), eavesdropping (with secure routing paths), and packet spoofing (blocking unknown sources) by unauthorized entities---a significant improvement over prior work (see Table~\ref{tab:sdn_defense}). We note that compromised nodes may still be able to probe or eavesdrop nodes for which they have sufficient security level and categories; however, their capabilities are significantly restricted to only that allowed by policy, and they can quickly be quarantined (via flow rules) upon detection.

}
}

\section{MLSNet Overview}
\label{sec:defenseapproach}
{\color{black}
In this section, we present our threat model, lattice of security labels, and policy constraints for MLSNet.

\subsection{Threat Model and Assumptions}
\label{sec:threatmodel}
For the assignment of security labels, we assume a \textit{Network Security Officer (NSO)}, as defined by \textit{Lu et al.}~\cite{lu1990model}, who assigns appropriate security labels (i.e., levels and categories) to the network entities (e.g., endpoint devices and forwarding nodes). The assignment can be done by leveraging the controller as it has a global view of the network, and it must be based on a security assessment of the entities in the network. For example, endpoints with unpatched operating systems should be assigned a lower security level, as they are more likely to contain vulnerabilities than others with the latest software updates. IoT devices or forwarding nodes connected to third-party networks can also be considered less secure, and therefore should be assigned a lower security level and a restricted set of categories. In contrast, endpoints containing more sensitive (e.g., top-secret) data should have a higher security level assigned since information flow to nodes with lower levels should be prevented.

\textit{MLSNet} aims to protect confidentiality by preventing leakage to unauthorized entities. We assume that nodes connected to a network may become compromised and have malicious intentions. In this scenario, we aim to limit an adversary's ability to further compromise the network.

Additionally, we assume the controller has an accurate view of the topology (i.e., nodes have not been spoofed). MLS cannot detect all forms of packet spoofing, and we rely on other SDN-based defenses to detect packet spoofing against the topology discovery service~\cite{hong2015poisoning}.

\subsection{Multilevel Security Lattices for Computer Networks}
\label{sec:sdnlattices}
To compute an SDN-based network configuration (set of flow rules installed to SDN switches) that satisfies the security policy, we must first consider security levels and categories. As drawn from \textit{Denning}~\cite{denning1976lattice}, we order the security levels used in our model according to the following: \textit{TopSecret (4) $>$ Secret (3) $>$ \textit{Confidential (2)} $>$ Public (1)}. For an SDN, we define the security categories as the packet types supported by the OpenFlow~\cite{openflow} protocol for matching incoming packets to flow rules: TCP, ICMP, etc. Although, any number of levels and categories can be defined to separate classes of flows; we just use the above descriptions as one example for evaluation.

The combination of a level and one or more categories then forms the label at a node. These labels form a lattice, a partially ordered set that reflects the secrecy and privilege requirements of communication in the network. We consider a node's label to be higher than another node's if the former's level is greater than or equal to, and the categories form a superset of, the latter's. This can be seen in the sample lattice shown in Figure~\ref{fig:accesscontrollattice1}. We note that this construction may lead to incomparable labels, where neither label is a subset/superset of the other, in which case communication would be denied by default. This will preserve confidentiality but with the caveat that not every flow may be accommodated.

Given the labels, the controller will install flow rules to the SDN switches to allow communication only if the security constraints are satisfied.
}

\begin{figure}[t]
	\footnotesize
	\centering
	\includegraphics[width=0.49\textwidth]{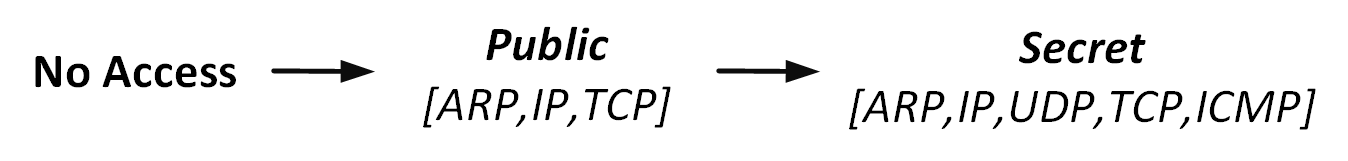}
	\small
	\caption{OpenFlow specific security lattice for networks used in our evaluations}
	\label{fig:accesscontrollattice1}
\end{figure}

\subsection{Security Policy Constraints}
\label{sec:policyconstraints}
In this section, we discuss the \textit{access control} and \textit{flow control} constraints, which form the basis of our security policy.

\textbf{Access Control.} As the first step to compute a security policy compliant network configuration, we determine if a subject (e.g., user or process) initiating communication with an object (e.g., file or resource) is allowed to exchange information with the object based on the security levels and categories. To define the constraints for access control, we have to consider if the subject is communicating with a receiver object (i.e., object receives from subject) or a provider object (i.e., subject sends to object). If the subject communicates with a provider object, then information flows from object $o$ to subject $s$; inversely, if object $o$ is a receiver object, then information flows from $s$ to $o$. In case the object is both a provider and receiver object at the same time, information flow between $s$ and $o$ is bidirectional. 

Considering these three cases, given security level $\sigma$ and categories $C$ of subject $s$ and object $o$, the authorized information flows are defined by a conventional MLS confidentiality model~\cite{blp76}. For a subject $s$ (e.g., workstation user) communicating with a provider object $o$, (e.g., mail server in the DMZ in Figure~\ref{fig:mls_scenario_1}), the following constraint must be satisfied:
\begin{equation}
\label{equ:accesscontrol_provobj}
\sigma_o \leq \sigma_s \text{ and } C_o \subseteq C_s, \forall o \in O, s \in S
\end{equation}
Secondly, for a subject $s$ (e.g., Wi-Fi client) communicating with a receiver object $o$ (e.g., network printer), the following constraint must be satisfied:
\begin{equation}
\label{equ:accesscontrol_recvobj}
\sigma_o \geq \sigma_s \text{ and } C_o \supseteq C_s, \forall o \in O, s \in S
\end{equation}
And for a subject $s$ (e.g., workstation user) communicating with an object $o$ that is both a provider and receiver (e.g., git repository), the following constraint must be satisfied:
\begin{equation}
\label{equ:accesscontrol_bothobj}
\sigma_o = \sigma_s \text{ and } C_o = C_s, \forall o \in O, s \in S
\end{equation}
\noindent 
Upon the initial arrival of a flow at the SDN controller from a subject, it can be determined if the subject $s$ is allowed to exchange information with object $o$ by considering the access control constraints. While we define our framework in a general way, security levels and categories are defined specific to SDNs, as discussed in Section~\ref{sec:sdnlattices}.

\textbf{Flow Control.} If the access control constraints are satisfied, information is allowed to flow between object $o$ and subject $s$. The next step before rule installation is for the controller to determine whether there exists a path between $o$ and $s$ such that the \textit{security level} of any forwarding node on the path between $o$ and $s$ is not lower than that of the flow. In Figure~\ref{fig:mls_scenario_1}, the security level of the switch connecting the publicly accessible Wi-Fi to the network is lower compared to the switches connected with the secure VPN, which are behind firewalls and only for internal users. Traversing lower classified nodes puts a flow at risk of being leaked to untrustworthy entities, being modified, or otherwise disrupted. Thus, protection of confidentiality is constrained by secure path selection, and MLSNet will choose an optimal path (if one exists) that satisfies this constraint for any candidate flow.

We can formulate such a constraint by stating that the security level $\sigma$ of a node $j$ on the path between $o$ and $s$ cannot be lower than the security level of the originating node of the flow. As with the access control constraint, we have to take into account whether a subject is communicating with a provider object, a receiver object, or an object that is both. If the subject is communicating with a provider object, then the following constraint must be satisfied:
\begin{equation}
\label{equ:flowcontrol_provobj}
\sigma_o \leq \sigma_j, \forall j\in V \text { on path } (o,s) \text{ for flow } f\in F
\end{equation}
Secondly, if the subject is communicating with a receiver object, then traffic is flowing from the subject toward the object, and the following constraint must be satisfied:
\begin{equation}
\label{equ:flowcontrol_recvobj}
\sigma_s \leq \sigma_j, \forall j\in V \text { on path } (o,s) \text{ for flow } f\in F
\end{equation}
Lastly, the access control constraint for communicating with an object that is a receiver and provider at the same time defines that $s$ and $o$ are required to have the same security level as stated in (\ref{equ:accesscontrol_bothobj}). Therefore, the flow control for such a case requires a forwarder node to have a security level that is higher or equal compared to the level of $s$ and $o$:
\begin{equation}
\label{equ:flowcontrol_bothobj}
(\sigma_s, \sigma_o)  \leq \sigma_j, \forall j\in V \text { on path } (o,s) \text{ for flow } f\in F
\end{equation}

In addition to the security labels, we also must consider the capacity $\kappa_{ij}$ of a link $(i,j)$ on a path between $s$ and $o$ for a flow with a size of $d^f$: $\kappa_{ij} \geq d^f \forall (i,j) \text{ on path } (o,s)$. As a trade-off for providing flow control, policy compliant paths may be longer than a shortest available path which does not consider a security policy. Additionally, in case two nodes satisfy the the access control constraint, there is no guarantee that a path between the nodes can be found which fulfills the flow control constraint. If such, there may be a path traversing nodes which \textit{do not} have a high enough security label. We refer to such cases as \textit{policy conflicts}. In Section~\ref{sec:secureproblem}, we present a model to minimize policy conflicts on flow paths. In short, it will find the best fitting configuration and report the exact locations on paths where policy conflicts exist. By deploying additional security mechanisms, such as \textit{declassification} via encrypted communication channels, such conflicts can be resolved, as we further discuss in Section~\ref{sec:declassification}.

\section{Policy Compliant Flows}
\label{sec:secpreservingsdn}
Given the policy constraints and security labels, we introduce optimization models to compute a flow-rule-based network configuration under consideration of policy compliance and resource availability. We first introduce an integer linear programming (ILP) model to maximize the number of flows strictly satisfying all security constraints. If no path meeting the required security constraints can be found for a flow $f$, the model will suggest to drop $f$. Further, we propose a second ILP-based optimization model which minimizes the sum of policy conflict values, given a topology and security labels, under the assumption that all flows $f\in F$ permitted by the access control constraint must be accommodated.

\subsection{Policy Compliant Flow Maximization Problem}
\label{sec:secureflowsstrict}
In this section, we introduce an ILP to maximize the number of flows that can be accommodated by a network under consideration of capacity and security constraints. We refer to this problem as the \textit{policy compliant flow maximization problem}, and formulate the constraints in (7). The optimization model shown determines if a network configuration fulfilling the defined security policy can be found to route the flows $F$ between the subjects $S$ and objects $O$. To compute a path, we first introduce a binary decision variable $x_{ij}^f$ to indicate if link $(i,j)\in E$ is used on the path for flow $f$ (i.e., $x_{ij}^f=1$) or not (i.e., $x_{ij}^f=0$). To decide if a flow $f$ can be accommodated, we also introduce the binary decision variable $\alpha^f$.
\newline \indent
In~\ref{con:sf_srcflow}, we add $\alpha^f$ to the link indication variable $x_{is}^f$ to trigger a flow $f$ at a subject $s$. To compute a path between the subject node $s$ and object node $o$, a flow $f$ is consumed at a node $o$, as stated in constraint~\ref{con:sf_dstflow}, by subtracting $\alpha^f$ from the link indication variable. In~\ref{con:sf_flowpres}, we state the flow preservation constraint to ensure that the sum of incoming flows into a node equals the sum of outgoing flows of a node. 

{
\begin{subequations}
	\label{equ:max_sec_flow}
	\footnotesize
	\begin{alignat}{3}
	& \max\ \sum_{f\in F} \alpha^f \label{con:sf_objective} &\\
	& \mbox{ s.t. } \nonumber& \\
	& \sum_{i:(i,s)\in E} x^f_{is} + \alpha^f = \sum_{j:(s,j)\in E} x^f_{sj}, & \forall f\in F& \label{con:sf_srcflow}\\
	& \sum_{j:(o,j)\in E} x^f_{oj} - \alpha^f = 0, & \forall f\in F \label{con:sf_dstflow}\\
	& \sum_{\enspace\enspace\thinspace i,j\in E} x^f_{ij} = \sum_{j,k\in E} x^f_{jk}, & \forall f\in F \label{con:sf_flowpres}\\
	& \sum_{\thinspace\thinspace i:(i,j)\in E} x_{ij}^f \leq 1, & \forall f\in F, j\in V\label{con:sf_xone}\\
	& \sum_{\quad\enspace\enspace f\in F} x^f_{ij}\cdot d^f \leq \kappa_{ij}, & \forall (i,j) \in E \label{con:sf_capacity}\\
	& \quad \alpha^f\cdot lev(\sigma_o,\sigma_s) = \alpha^f, & \hspace{-25mm} \forall o(f)\in O, s(f)\in S, f\in F\label{con:sf_access} \\
	& \quad \alpha^f\cdot cat(\lambda_o^c,\lambda_s^c) = \alpha^f,  &\nonumber\\
	& &\hspace{-5mm} \forall o(f)\in O, s(f)\in S, c\in C, f\in F \label{con:sf_seccat} \\
	& \quad x_{ij}^f\cdot orig(\sigma_o,\sigma_s) \leq x_{ij}^f\cdot \sigma_j, & \nonumber\\
	& & \hspace{-5mm} \forall (i,j) \in E, o(f)\in O, s(f)\in S, f\in F\label{con:sf_secclass} \\
	& \quad x^f_{ij} \in \{0,\: 1\}, & \forall (i,j)\in E, f\in F \nonumber\\
	& \quad \alpha^f \in \{0,\: 1\}, & \forall f\in F \nonumber
	\end{alignat}
\end{subequations}}

We add constraint~\ref{con:sf_xone} to limit the number of visits of a node to one for each flow.
Constraint~\ref{con:sf_capacity} ensures that the given capacity $\kappa_{ij}$ of a link $(i,j)\in E$ is not exceeded for forwarding flows over a link $i,j$ with a size of $d^f$ per flow. Typically, in bidirectional communication in computer networks the size of the request flow is different than the size of the reply flow. Since we assume symmetric routes, the flow size variable $d^f$ should be chosen to account for the flow size in both directions. Additionally, since new flow demands typically arrive at different times in a network, we can replace the above link capacity $\kappa_{ij}$ with the residual capacity $\widetilde{\kappa}_{ij}$ which states the remaining capacity on a link $(i,j)\in E$ considering the existing flows in a network traversing link $(i,j)$.

In constraints~\ref{con:sf_access} and~\ref{con:sf_seccat}, we define the access control properties. Constraint~\ref{con:sf_access} ensures that a flow $f$ between a subject $s$ and an object $o$ is only permitted if the function $lev(\sigma_o,\sigma_s)$, shown in (\ref{equ:seclevfct}), returns 1, indicating that the security levels of $s$ and $o$ allow communication:
\begin{equation}
\label{equ:seclevfct}
lev(\sigma_o, \sigma_s)=
\begin{cases}
  1, & \text{\textbf{if} o is provider object \textbf{and} } \sigma_o \leq \sigma_s  \\
  1, & \text{\textbf{if} o is receiver object \textbf{and} } \sigma_o \geq \sigma_s  \\
  1, & \text{\textbf{if} o is both \textbf{and} } \sigma_o = \sigma_s  \\
  0, & \text{otherwise}  \\
\end{cases}
\end{equation}

As defined in Section~\ref{sec:policyconstraints} for access control, we further have to ensure that the subject $s$ and object $o$ have the appropriate security categories before calculating a path. In function $cat(\lambda_o^c,\lambda_s^c)$ shown in (\ref{equ:seccatfct}), we model the requirement of security categories to allow a flow between $s$ and $o$:
\begin{equation}
\label{equ:seccatfct}
cat(\lambda_o^c,\lambda_s^c)=
\begin{cases}
  1 - (\lambda_o^c-\lambda_o^c\cdot\lambda_s^c), &\hspace{-2mm} \text{\textbf{if} o is provider}  \\
  1 - (\lambda_s^c-\lambda_s^c\cdot\lambda_o^c), &\hspace{-2mm} \text{\textbf{if} o is receiver}  \\
  1 - (\lambda_s^c - \lambda_o^c)\cdot(\lambda_s^c - \lambda_o^c) , &\hspace{-2mm} \text{\textbf{if} o is both}  \\
  0, &\hspace{-2mm} \text{otherwise}  \\
\end{cases}
\end{equation}
To mathematically define this, we introduce variable $\lambda_i^c$ which indicates if a node $i$ has a security category $c$, i.e., $\lambda_i^c=1$, or not, i.e., $\lambda_i^c=0$. As an example for the operation of function $cat()$, suppose a subject $s$ wants to communicate with a provider object $o$. In order to permit the flow, the constraint that $C_o \subseteq C_s$ must be satisfied. To evaluate if the security categories $C_o$ of an object are a subset of the categories in $C_s$, we introduce the formulation $1 - (\lambda_o^c-\lambda_o^c\cdot\lambda_s^c)$ as shown in (\ref{equ:seccatfct}). This will evaluate to $0$ if subject $s$ does not have a security category $c$, but object $o$ does, i.e., $(1 - (1-1\cdot 0)) = 0$. Such a case does not fulfill the access control constraint, and therefore the flow cannot be admitted, i.e., $\alpha_f=0$.

Function $cat()$ works in a similar way if $o$ is a receiver object. In case $o$ is both a provider and receiver, function $cat()$ evaluates to $1$ if $C_o = C_s$. As stated in constraint~\ref{con:sf_seccat}, the function $cat()$ has to return $1$ for all categories $c\in C$ for a flow $f$ between a subject $s(f)\in S$ and an object node $o(f)\in O$. Additionally, in constraint~\ref{con:sf_secclass}, we define the secure flow property to prevent information flow to lower classified nodes. Thus, for each next node $j$ on a link $(i,j)$ of a flow $f$, indicated by the decision variable $x_{ij}^f$, the security class of the originating node of flow $f$ (i.e., the subject if the object is a receiver, and the object otherwise) has to be less or equal to the security class at the next node $j$ on the path. We define function $orig(\sigma_o,\sigma_s)$ as shown in (\ref{equ:secfct}), where $orig()$ returns the security level depending on the type of object node $o$, according to the defined flow control constraint in Section~\ref{sec:policyconstraints}:

\begin{equation}
\label{equ:secfct}
orig(\sigma_o,\sigma_s)=
\begin{cases}
  \sigma_s, & \text{\textbf{if} o is receiver}  \\
  \sigma_o, & \text{otherwise}
\end{cases}
\end{equation}

\noindent This last constraint ensures that on a path between a subject $s$ and an object node $o$, no forwarding nodes with a lower security level compared to the security level of the originating node of the flow are visited. We then use the specified constraints~\ref{con:sf_srcflow}-\ref{con:sf_secclass} as the basis for our heuristic-based maximization algorithm discussed in the next section.

\subsection{Policy Compliant Flow Maximization Algorithm}
\label{sec:securealgorithm}
The linear programming model introduced in Section~\ref{sec:secureflowsstrict} is a \textit{binary integer programming} model, a special case of integer linear programming (ILP) since all variables are binary. Integer linear programming models are NP-hard problems in general, and the special case of binary integer programming is one of Karp's 21 NP-complete problems~\cite{karp1972reducibility}. Although solvers such as \textit{Gurobi}~\cite{gurobi} are efficient in computing a solution for such problems, binary integer programming models can be impractical to solve for certain inputs.

\begin{algorithm}[t]
	\caption{PolicyCompliantPath($G$,$s$,$o$,$d^f$)}
	\label{alg:secpathmax}
	\begin{algorithmic}[1]
		\STATE $V$ = nodes in $G$
		\FORALL{$v \in V$} 
		\STATE $dist[v]$ = infinity, $prev[v]$ = null
		\ENDFOR
		\STATE $dist[s]$ = $0$
		\STATE $N$ = nodes in $G$ 
		\IF{$lev(\sigma_o,\sigma_s)=1$ and $cat(\lambda_o^c,\lambda_s^c)=1, \forall \lambda_o^c\in C_o, \lambda_s^c\in C_s$} \label{alg:accesscon}
		\WHILE{$N$ not empty}
		\STATE $i$ = node in $N$ with smallest $dist[]$
		\STATE remove $i$ from $N$
		\FORALL{adjacent node $j$ of $i$}
		\IF {$orig(\sigma_o,\sigma_s) \leq \sigma_j$ and $d^f\leq \widetilde{\kappa}_{ij}$} \label{alg:flowseccon}
		\STATE $dist_{new}$ = $dist[i] + 1$
		\IF {$dist_{new} \leq dist[j]$} \label{alg:mindist}
		\STATE $dist[j]$ = $dist_{new}$
		\STATE $prev[j]$ = $i$
		\ENDIF 
		\ENDIF
		\ENDFOR
		\ENDWHILE
		\ENDIF 
		\STATE \textbf{return} $prev$
	\end{algorithmic}
\end{algorithm}

To address this issue, we also formulate a heuristic algorithm to compute a security compliant path between subjects and objects based on a modification of \textit{Djikstra's} shortest path algorithm. Algorithm~\ref{alg:secpathmax} is a modification of \textit{Dijkstra's} shortest path algorithm where we add the access control and secure flow control constraints, similar to the constraints presented in (\ref{equ:max_sec_flow}). Specifically, we formulate the access control constraint in line~\ref{alg:accesscon} based on the introduced functions $lev()$ as defined in (\ref{equ:seclevfct}) and $cat()$ as defined in (\ref{equ:seccatfct}). To compute a secure path between $s$ and $o$ we define the constraints in line~\ref{alg:flowseccon} to only consider an adjacent node $j$ of a link if the security level of node $j$ is greater or equal the security level of the originating node of flow $f$ and the link connecting node $i$ and $j$ has enough residual capacity to accommodate flow $f$. In effect, the introduced model and algorithm will compute paths that accommodate the maximum number of flows $f\in F$ between a subject node $s$ and an object node $o$, with consideration for security and link capacity.

\subsection{Policy Conflict Minimization Model}
\label{sec:secureproblem}
Finding a path fulfilling all security conditions might not always be possible considering the nature of real-world networks. In contrast to the previous model, here we assume that \textit{all} flows fulfilling the access control and link capacity constraints must be accommodated in the network, which may lead to policy conflicts. Policy conflicts are conditions where a flow is visiting a node on a path that has a lower security level than the transferred information (i.e., than the sender node), and we quantify a policy conflict as the numerical difference between those security levels. Considering the lattice in Section~\ref{sec:sdnlattices}, we assume that nodes classified as \textit{Confidential (2)} have a higher risk of being compromised than nodes classified as \textit{Secret (3)}. The goal here is to minimize policy conflicts; therefore, if information classified as \textit{Top Secret (4)} is transferred on a path with policy conflicts, it is preferable to select nodes with the smallest numerical difference (i.e., \textit{Secret (3)} nodes are preferred over \textit{Confidential (2)} nodes). 

Resolving policy conflicts requires additional security measures (e.g., \textit{declassification}). 
The larger a policy conflict (i.e., higher numerical difference in security levels), the more an additional security measure will cost, in terms of transmission time or computation overhead. By minimizing the numerical distance of policy conflicts, we aim to minimize the \textit{cost} required to apply additional security measures to meet a defined security policy.

To achieve this, we compute a network configuration in a two-step process. We first select a subset of the flows $F_l \subseteq F$ that fulfill the access control constraints, and second, compute paths between subjects and objects with the objective to minimize security policy conflicts. We define the access control constraints as follows:
\begin{equation}
\begin{split}
\label{equ:flowaccess}
F_l = & \{ f\in F: lev(\sigma_o,\sigma_s)=1 \text{ and } cat(\lambda_o^c,\lambda_s^c)=1, \\
& \forall \lambda_o^c\in C_o, \forall \lambda_s^c\in C_s, o \in O, s \in S\}
\end{split}
\end{equation}

Next, for the set of legitimate flows $F_l$, we also define an objective function, \textit{conf}, to find a network configuration that accommodates all flows in $F_l$ while minimizing the policy conflicts on a path of a flow $f\in F_l$ between a subject $s$ and an object $o$. The function returns the difference between the security level of the flow's originating node, given by $orig(\sigma_o,\sigma_s)$, and the security level $\sigma_j$ of a node $j$ on the path between $s$ and $o$ if $\sigma_j < orig(\sigma_o,\sigma_s)$. More formally:
\begin{equation}
\label{equ:viofct}
conf(\sigma_o,\sigma_s,\sigma_j)=
\begin{cases}
  orig(\sigma_o,\sigma_s) - \sigma_j, & \text{\textbf{if}}\ \sigma_j < orig(\sigma_o,\sigma_s)  \\
  0, & \text{otherwise}
\end{cases}
\end{equation}
We aim to minimize the policy conflicts caused by visited nodes with lower security levels. Assuming a flow originates from a node $o$, we define the severity of the policy conflict by the numerical distance from level $\sigma_o$ of node $o$ to a node $j$ with level $\sigma_j$, if $\sigma_j < \sigma_o$. Then, choosing a node $j$ over a node $h$, where $(\sigma_o-\sigma_j) < (\sigma_o - \sigma_h)$, is preferable. And for selecting a secure path, we want to give preference to these nodes with a smaller difference in security level with the originating node, even if such a path is significantly longer than the shortest path. To model this, we introduce a factor $\gamma$ and define our objective function as follows:
\begin{equation}
\label{equ:secdiffobj}
\min\:\sum_{f\in F_l}\sum_{i,j\in E} x^f_{ij}\cdot \gamma^{conf(\sigma_o(f),\sigma_s(f),\sigma_j)}
\end{equation}
In (\ref{equ:secdiffobj}), $x^f_{ij}$ denotes the decision variable if link $(i,j)$ is selected as part of the path between $s$ and $o$ for a flow $f$. In the objective function as shown in (\ref{equ:secdiffobj}), $\sigma_o(f)$ denotes the security level of object node $o$ of a flow $f$, $\sigma_s(f)$ denotes the security level of subject node $s$ of a flow $f$. The security level of a node $j$ on the path between $s$ and $o$ is defined by $\sigma_j$.
The factor $\gamma$ controls the length of a path that should be chosen as a trade-off for visiting nodes with a smaller distance in terms of security levels. We visualize this in an example shown in Figure~\ref{fig:pathselection}. 

\begin{figure}
	\centering
	\includegraphics[width=0.33\textwidth]{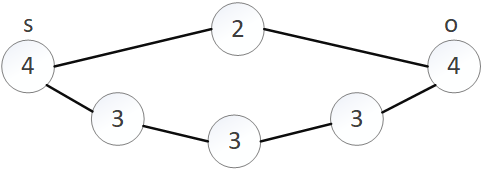}
	\caption{Path selection based on security levels.}
	\label{fig:pathselection}
\end{figure}
Considering this small network, two possible paths exist between $s$ and $o$. The numbers in the nodes indicate their security level. If we select a factor $\gamma=4$, the top path results in a value of $4^{4-2}=16$, while the bottom path has a \textit{smaller} policy conflict value of $3\cdot 4^{4-3}=12$, and thus would be selected. In contrast, if we select a smaller value for $\gamma$ (e.g., $\gamma=2$) then considering the top path, a transition from the node with security level $4$ to the node with security level $2$ has to be made, resulting in a policy conflict value of $2^{4-2}=4$. Computing the policy conflict value for the bottom path would result in $3\cdot 2^{4-3}=6$, since we have to visit three nodes with a difference in the security level of $1$. Based on this, the top path would be selected, although from a security perspective, the bottom path may be more preferable since nodes with a smaller security level difference are visited. This example shows that the factor $\gamma$ controls the selection of longer paths visiting nodes with a smaller security level difference. To always select paths with the smallest security level difference, $\gamma$ can be set to the network diameter $+$ $1$, in terms of hop count, which we prove as follows.
\begin{lemma}
\label{lem:secpath}
To select a longer path with lower policy conflicts, we must set $\gamma$ to the maximum path length $+$ $1$. 
\end{lemma}
\begin{proof}
Assuming a flow originating from an object $o$, we have to show that for a path of flow $f$ defined by a set of links $(i,j)\in E$ indicated by the decision variable $x_{ij}^f$, the value of policy conflicts specified as $\sum_{i,j\in E} x^f_{ij}\cdot \gamma^{\sigma_o - \sigma_j}$ and assuming that $\forall j, \sigma_j < \sigma_o$, is larger for a path with higher policy conflicts than a potentially longer path with a lower conflict value if $\gamma$ is chosen appropriately. Given a candidate node for the path of flow $f$ with a policy conflict of $a$, we want to choose $\gamma$ so that a potentially longer path $y$ over a set of nodes with a lower policy conflict of $b$ is selected, i.e., $\gamma^{a} > y\cdot \gamma^{b}$. Since $a > b$, we can express $b$ as $a - q$, where $q$ is the numerical difference of the security levels of $a$ and $b$, i.e., $q=a - b$. By replacing $b$ with $a - q$, we can write the inequality above as $\gamma^{a}\cdot \gamma^{q} > y\cdot \gamma^{a}$. Assuming the smallest absolute difference of two unequal security classes, i.e., $q=1$, the inequality above can be written as $\gamma > y$. Therefore, we can say that in order to select a path that is $y$ hops longer, over a set of nodes with a lower policy conflict value, we have to select a value for $\gamma$ that is at least $y+1$. This also holds for larger security differences, since $\gamma^q < \gamma^{q+1}$ holds true for positive values of $q$.
\end{proof}

\subsection{Policy Conflict Minimization Problem}
\label{sec:secureflows}
To minimize the policy conflicts on a path, we formulate the optimization problem as an integer linear programming (ILP) model. We refer to this problem as the \textit{security policy conflict minimization problem}, as shown in (\ref{equ:optmial_net_sec}).

To trigger a flow at a node $s$, \ref{con:srcflow} adds $1$ to the decision variable $x_{is}^f$. In our formulation to compute a path from $s$ to $o$, the flow is consumed at node $o$ as stated in constraint~\ref{con:dstflow}. In~\ref{con:flowpres}, we state the flow preservation constraint to ensure that the sum of incoming flows to a node equals the sum of outgoing flows. Constraint~\ref{con:capacity} ensures that the given capacity $\kappa_{ij}$ of a link $(i,j)\in E$ is not exceeded by forwarding flows $f\in F_l$ with a size of $d^f$ per flow. In~\ref{con:capacity}, we assume that the flow size $d^f$ is chosen to include traffic between $s$ and $o$ in both directions since we assume symmetric routes. Since new flow demands typically arrive at different times in a network, we can replace the above link capacity $\kappa_{ij}$ with the residual capacity $\widetilde{\kappa}_{ij}$ which states the remaining capacity on a link $(i,j)\in E$ considering the existing flows traversing link $(i,j)$. Accordingly, with the specified constraints~\ref{con:srcflow}-\ref{con:capacity}, the introduced model will compute a path for every flow $f\in F_l$ between a subject node $s(f)$ and an object node $o(f)$ with the objective function as defined in \ref{con:objective}.

\begin{subequations}
	\label{equ:optmial_net_sec}
	\footnotesize
	\begin{alignat}{3}
	& \min\:\sum_{f\in F_l} \sum_{i,j\in E} x^f_{ij}\cdot \gamma^{conf(\sigma_{o}(f),\sigma_{s}(f),\sigma_j)} \label{con:objective} &\\
	& \mbox{ s.t. } \nonumber& \\
	& \sum_{i:(i,s)\in E} x^f_{is} + 1 = \sum_{j:(s,j)\in E} x^f_{sj}, & \forall f\in F_l & \label{con:srcflow}\\
	& \sum_{j:(o,j)\in E} x^f_{oj} - 1 = 0, & \forall f\in F_l \label{con:dstflow}\\
	& \enspace\thinspace \sum_{i,j\in E} x^f_{ij} = \sum_{j,k\in E} x^f_{jk}, & \forall f\in F_l \label{con:flowpres}\\
	& \enspace\thinspace\thinspace \sum_{f\in F_l} x^f_{ij}\cdot d^f \leq \kappa_{ij}, & \forall (i,j) \in E \label{con:capacity}\\
	& \enspace\enspace x^f_{ij} \in \{0,\: 1\}, & \forall (i,j)\in E, f\in F_l \nonumber
	\end{alignat}
\end{subequations}

\subsection{Policy Conflict Minimization Algorithm}
\label{sec:minsecvioalgo}
As discussed in Section~\ref{sec:securealgorithm}, ILP models with binary integer variables, such as (\ref{equ:optmial_net_sec}), are typically NP-hard and can be impractical to solve for certain input sequences. To address this, we also propose a heuristic algorithm to approximate an optimal solution and replace the objective to find the shortest path with the objective to compute a path with the smallest sum of policy conflict values (Algorithm~\ref{alg:secviomin}, lines~\ref{alg:objmin}-\ref{alg:flowsecconmin}).
\newline \indent
Since Algorithms \ref{alg:secpathmax} and \ref{alg:secviomin} are based on \textit{Dijkstra's shortest path algorithm}, we can express their time complexity as $O(|F|\cdot(|E|+|V|log|V|))$ for a number of $|F|$ flows.

\begin{algorithm}[t]
	\caption{MinConflictPath($G$,$s$,$o$,$d^f$)}
	\label{alg:secviomin}
	\begin{algorithmic}[1]
		\STATE $V$ = nodes in $G$
		\FORALL{$v \in V$} 
		\STATE $conf[v]$ = infinity, $prev[v]$ = null
		\ENDFOR
		\STATE $conf[s]$ = $0$, $N$ = nodes in $G$ 
		\IF{$lev(\sigma_o,\sigma_s)=1$ and $cat(\lambda_o^c,\lambda_s^c)=1, \forall \lambda_o^c\in C_o, \lambda_s^c\in C_s$} \label{alg:accessmin}
		\WHILE{$N$ not empty}
		\STATE $i$ = node in $N$ with smallest $conf[]$
		\STATE remove $i$ from $N$
		\FORALL{adjacent node $j$ of $i$}
		\IF {$d^f\leq \widetilde{\kappa}_{ij}$} \label{alg:capamin}
		\STATE $conf_{new}$ = $conf[j] + \gamma^{conf(\sigma_o,\sigma_s,\sigma_j)}$ \label{alg:objmin}
		\IF {$conf_{new} \leq conf[j]$} \label{alg:flowsecconmin}
		\STATE $conf[j]$ = $conf_{new}$
		\STATE $prev[j]$ = $i$
		\ENDIF 
		\ENDIF
		\ENDFOR
		\ENDWHILE
		\ENDIF 
		\STATE \textbf{return} $prev$
	\end{algorithmic}
\end{algorithm}

\subsection{Resolving Policy Conflicts}
\label{sec:declassification}
As we've shown in the previous section, paths with sufficient security levels and capacity may not always exist for two nodes permitted to communicate. In such a case, additional security mechanisms must be applied on the flow in order to be policy compliant. These mechanisms typically involve a cost to implement (e.g., increased transmission delay, processing time, or capacity), which the latter model aims to minimize.

An important mechanism for resolving policy conflicts is \textit{declassification}, which is the process of lowering the security level of the information. \textit{Sabelfeld et al.}~\cite{sabelfeld2009declassification} discuss a general framework for declassification by defining the dimensions of information release, including: \textit{what} information is released, \textit{who} releases the information, \textit{where} information is released, and \textit{when} it is released. By analyzing these dimensions, a 
\noindent network operator is then able to evaluate the risks and benefits of declassification to resolve certain security policy conflicts.
In our framework, declassification can involve lowering a flow's security level so it can traverse a path with lower classified nodes than the originating node. Methods to achieve this include traffic camouflaging techniques \cite{cai2014cs} or VPNs to defend against traffic analysis. Resolving policy conflicts can also be achieved by the NSO relabeling certain nodes in the network (e.g., increasing the security level of forwarding nodes) after upgrading the security measures on a switch and re-evaluating its security level. Our proposed optimization model to minimize policy conflicts will point out exactly which components of the network topology are causing conflicts; therefore, relabeling of nodes can be a permanent solution to policy conflicts which may reoccur.

\section{Secure Flow Rule Construction}
\label{sec:secureflowrules}
To realize a policy compliant network configuration, in the following we define a set of principles for the construction of secure flow rules which preserve confidentiality.

\subsection{Isolating Flows}
\label{sec:sdn_iso_flows}
Attacks that exploit the composition of flow rules are effective because the matching criteria often only identifies packets by a limited set of header fields, as discussed by \textit{Achleitner et al.}~\cite{achleitner2017adversarial}. If the flow rules are only matching packets against header fields of a certain network layer (e.g., IP addresses), then the information in other layers will be seen as \textit{"wild cards"} and thus will be ignored. This problem motivates the construction of SDN flow rules with consideration of information spanning all network layers.

The OpenFlow protocol~\cite{openflow} defines a set of matching fields supporting different network layers. Multiple endpoints may share lower layer fields such as physical ingress port; thus, to differentiate them and identify their security levels, we must include fields from higher network layers (e.g., IP or Ethernet addresses). But security leaks caused by the exchange of certain packet types in SDN-enabled networks~\cite{hong2015poisoning,dhawan2015sphinx,achleitner2017adversarial} motivate the use of categories in a security lattice to offer finer granularity of information exchange in SDN flow rules. Therefore, we derive these categories from additional packet header fields (e.g., $ARP$, $IP$, $TCP$, $UDP$ and $ICMP$), and use them in enforcing the security policy.

Additionally, with this general framework, a security category can be defined with even finer granularity. For example, by specifying field subtypes: \textit{ICMP type 8 code 0}, to allow ping packets. Thus, the various header fields allow greater flexibility when defining the security policy, and unlike traditional networks, the policy can be efficiently managed by sending \textit{flow\_mod} messages to the forwarding nodes to update their routing tables.

\begin{algorithm}[t]
	\caption{GenerateFlowRule($R$,$A$,$P$,$next$)}
	\label{alg:genrulealgo}
	\begin{algorithmic}[1]
		\STATE $rule.append("match:")$ \label{genrulealgo:match}
		\FORALL{$r\in R$} 
		\IF {$r\in P$} \label{genrulealgo:matchfields}
		\STATE $rule.append(r=P(r))$ \label{genrulealgo:addmatchfield}
		\ENDIF
		\ENDFOR
		\STATE $rule.append("action:")$ \label{genrulealgo:actions}
		\IF {$next$ != $drop$} \label{genrulealgo:deny}
		\FORALL{$a\in A$} 
		\STATE $rule.append(a)$ \label{genrulealgo:actionfields}
		\ENDFOR
		\STATE $rule.append("next")$ \label{genrulealgo:output}
		\ELSE
		\STATE $rule.append("drop")$
		\ENDIF
		\STATE \textbf{return} $rule$
	\end{algorithmic}
\end{algorithm}

\subsection{Constructing Secure Flow Rules}
\label{sec:sdn_sec_rules}
In SDN-enabled networks, we must consider the assigned security labels during the construction of flow rules at the controller. As described previously, security leaks can arise with imprecise matching criteria. Considering this, we must construct precise flow rules which ensure that only packets fulfilling the defined security level and category constraints can be transmitted. More formally, for secure rule construction, we represent the set of fields supported for network layer $N_i\in N$, where $N$ is the set of all network layers, as $R_{Ni}$. Then, the superset of fields to be matched against some packet $P$ during secure rule construction can be realized by taking the union of sets, which we denote as $R$:
\begin{equation}
\label{equ:ruleinfoset}
R = \bigcup\limits_{N_i \in N : \forall N_i\leq N_P} R_{N_i}
\end{equation}
This formulation ensures that---after satisfying access control constraints---a flow is \textit{isolated} and handled correctly according to its security categories (i.e., packet fields).

Besides matching criteria, OpenFlow also defines action sets, specifying what actions to apply on matched packets. These include a required action part (e.g., forwarding or dropping) and optional actions (e.g., rewriting packet header fields). We denote the required action part as $next$, which specifies to either send a packet to a specific output port or to drop it. Similarly, we specify the set of optional actions, such as rewriting addresses, as the action set $A$.

Based on sets $R$ and $A$ for a packet $P$, we formulate Algorithm~\ref{alg:genrulealgo}. We begin rule construction by defining the \textit{matching} part of a flow rule in line~\ref{genrulealgo:match}. We continue to iterate through the set of matching fields in $R$, as defined in Equation~\ref{equ:ruleinfoset}, and check in line~\ref{genrulealgo:matchfields} if a specified field $r$ can be applied to a value in packet $P$. If this evaluates to true, we add the matching field $r$ and its associated value $P(r)$ to the flow rule in line~\ref{genrulealgo:addmatchfield}.
In line~\ref{genrulealgo:actions}, we add the action part of a flow rule and check in line~\ref{genrulealgo:deny} if the action is to drop the packet. In case we specify a rule to drop packets of a flow with specific protocol types, it must be ensured that the priority of that rule is higher than other rules for the flow which may allow forwarding for other protocol types (i.e., allowing most IP traffic, but disallowing any UDP over IP).

To resolve any flow rule conflicts between rule actions, we refer to existing frameworks such as \textit{Porras et al.}~\cite{porras2012security}. Nonetheless, if the packet is forwarded, then the set of action fields and the output port are added to the rule, as shown in lines~\ref{genrulealgo:actionfields} and \ref{genrulealgo:output}. With this construction, we are able to properly generate a secure flow rule configuration and isolate the flows to ensure that confidentiality of information flow is preserved in the network.

\section{Evaluation}
\label{sec:evaluation}
\begin{table*}
	\centering
	\caption{Flow maximization benchmark.}
	\label{tab:maxbenchmarknew}
    \begin{tabular}{|c|c|m{2.25em}|m{2.25em}||m{2.25em}|m{2.25em}|m{2.25em}||m{2.25em}|m{2.25em}|m{2.25em}||m{2.25em}|m{2.25em}|m{2.25em}|}
		\cline{2-13}
		\multicolumn{1}{c|}{} & \multicolumn{3}{c||}{AS network} & \multicolumn{3}{c||}{k=8} & \multicolumn{3}{c||}{k=12} & \multicolumn{3}{c|}{k=16} \\ \hline
		\# Lattice Levels & \makecell{2 lev.} & \makecell{3 lev.} & \makecell{4 lev.} & \makecell{2 lev.} & \makecell{3 lev.} & \makecell{4 lev.} & \makecell{2 lev.} & \makecell{3 lev.} & \makecell{4 lev.} & \makecell{2 lev.} & \makecell{3 lev.} & \makecell{4 lev.} \\ \hline
		LP (no congestion) & $85\%$ & $75\%$ & $65.7\%$ & $79.5\%$ & $64.6\%$ & $58.3\%$ & $79.5\%$ & $66.6\%$ & $59\%$ & $81\%$ & $68.4\%$ & $64\%$ \\ \hline
		HA (no congestion) & $79.6\%$ & $71\%$ & $63.1\%$ & $63\%$ & $52\%$ & $50.6\%$ & $66.5\%$ & $57\%$ & $51.1\%$ & $68.5\%$ & $61.1\%$ & $53\%$ \\ \hline
		LP (cong. network) & $53.4\%$ & $45.2\%$ & $37.4\%$ & $45.5\%$ & $39.3\%$ & $34.3\%$ & $56\%$ & $46.5\%$ & $36.5\%$ & $56\%$ & $50.6\%$ & $36.5\%$ \\ \hline
		HA (cong. network) & $51.9\%$ & $41.8\%$ & $35.8\%$ & $44.5\%$ & $33.5\%$ & $31.9\%$ & $50\%$ & $38\%$ & $32.3\%$ & $52.3\%$ & $45.8\%$ & $34.5\%$ \\ \hline
	\end{tabular}
\end{table*}
\begin{table*}
	\centering
	\caption{Policy conflict minimization benchmark.}
	\label{tab:minbenchmarknew}
    \begin{tabular}{|c|c|m{2.25em}|m{2.25em}|m{2.25em}||m{2.25em}|m{2.25em}|m{2.25em}||m{2.25em}|m{2.25em}|m{2.25em}||m{2.25em}|m{2.25em}|m{2.25em}|}
		\cline{3-14}
		\multicolumn{1}{c}{} & \multicolumn{1}{c|}{} & \multicolumn{3}{c||}{AS network} & \multicolumn{3}{c||}{k=8} & \multicolumn{3}{c||}{k=12} & \multicolumn{3}{c|}{k=16} \\ \cline{1-14} 
		\multicolumn{2}{|c|}{\# Lattice Levels} &
		\makecell{2 lev.} & \makecell{3 lev.} & \makecell{4 lev.} & \makecell{2 lev.} & \makecell{3 lev.} & \makecell{4 lev.} & \makecell{2 lev.} & \makecell{3 lev.} & \makecell{4 lev.} & \makecell{2 lev.} & \makecell{3 lev.} & \makecell{4 lev.} \\ \hline
		LP & no conflict & $85\%$ & $75\%$ & $65.7\%$ & $79.5\%$ & $64.6\%$ & $58.3\%$ & $79.5\%$ & $66.6\%$ & $59\%$ & $81\%$ & $68.4\%$ & $64\%$ \\ \cline{2-14}
		\multicolumn{1}{|c|}{} & 1 lev. diff. & $15\%$ & $21.5\%$ & $29.3\%$ & $20.5\%$ & $24.6\%$ & $24\%$ & $20.5\%$ & $26\%$ & $28.5\%$ & $19\%$ & $22\%$ & $20\%$ \\ \cline{2-14}
		\multicolumn{1}{|c|}{} & 2 lev. diff. & - & $3.5\%$ & $4.3\%$ & - & $10.8\%$ & $12.7\%$ & - & $7.4\%$ & $10\%$ & - & $9.6\%$ & $11\%$ \\ \cline{2-14}
		\multicolumn{1}{|c|}{} & 3 lev. diff. & - & - & $0.7\%$ & - & - & $5\%$ & - & - & $2.5\%$ & - & - & $5\%$ \\ \hline
		HA & no conflict & $79.6\%$ & $71\%$ & $63.1\%$ & $63\%$ & $52\%$ & $50.6\%$ & $66.5\%$ & $57\%$ & $51.1\%$ & $68.5\%$ & $61.1\%$ & $53\%$ \\ \cline{2-14}
		\multicolumn{1}{|c|}{} & 1 lev. diff. & $20.4\%$ & $21.5\%$ & $24.6\%$ & $37\%$ & $25.7\%$ & $22.1\%$ & $33.5\%$ & $29\%$ & $27.5\%$ & $31.5\%$ & $23.2\%$ & $23\%$ \\ \cline{2-14}
		\multicolumn{1}{|c|}{} & 2 lev. diff. & - & $7.5\%$ & $6.6\%$ & - & $22.3\%$ & $22.6\%$ & - & $14\%$ & $14.7\%$ & - & $15.7\%$ & $14\%$ \\ \cline{2-14}
		\multicolumn{1}{|c|}{} & 3 lev. diff. & - & - & $5.7\%$ & - & - & $4.7\%$ & - & - & $6.7\%$ & - & - & $10\%$ \\ \hline
	\end{tabular}
\end{table*}

{\color{black} 

With an MLS policy, adversarial capabilities (in terms of probing, eavesdropping, and lateral movement) are by definition restricted to only that allowed by policy. Here, we still want to be able to route all legitimate flows. In the following, we demonstrate that (in comparison to not enforcing a security policy) a network administrator can still provide strong coverage of network flows. The goal of our approach is to achieve this, while also reducing the security cost associated with guaranteeing all flows be routed. In Table~\ref{tab:maxbenchmarknew}, we report on the performance of our framework to find policy compliant paths for flows in various topologies and with lattices of different sizes. In Table~\ref{tab:minbenchmarknew}, we report MLSNet's performance when minimizing the policy conflicts, where accommodating the remaining flows may require routing along paths containing nodes with a lower security-level than required. We further show MLSNet's ability to mitigate common attacks (see Section~\ref{sec:securitydiscussion}), such as those executed by the recently proposed reconnaissance tool SDNMap~\cite{sdnmap, achleitner2017adversarial}.
}

\subsection{Accommodating Network Flows}
\label{sec:securitycost}
To evaluate the ability of our framework to maximize the number of policy compliant flows and minimize policy conflicts on paths, we test the introduced linear programming (LP) models and heuristic algorithms (HA) on different network topologies. We first consider a realistic autonomous system (AS) network topology. Then, to model common data-center and cloud topologies, we consider different \textit{k}-ary fat-tree networks~\cite{al2008scalable}, where \textit{k} is the port density of each switch in the network (e.g., 8, 12, and 16 ports).  We consider lattices of 2-4 security levels which are evenly distributed and randomly assigned to the nodes in a network. To generate flows, we randomly pick source and destination node pairs which fulfill the access control constraint and compute paths with our linear program models and heuristic algorithms. Further, we consider networks with different link capacities to simulate congestion. Our results for flow maximization and policy conflict minimization (for the remaining flows) are averaged over several runs and shown in Tables~\ref{tab:maxbenchmarknew} and~\ref{tab:minbenchmarknew}.

{\color{black} 
We explore the number of flows able to be routed in the AS$3257$ Rocketfuel~\cite{spring2002measuring_long} topology (161 nodes, 656 links), as well as \textit{8}-ary (208 nodes, 384 links), \textit{12}-ary (612 nodes, 1296 links), and \textit{16}-ary (1344 nodes, 3072 links) fat-tree networks. For flow maximization, shown in Table~\ref{tab:maxbenchmarknew}, our framework shows that a majority of flows was always routed securely. For any of the topologies, the number of flows routed reached, at a maximum, approximately $85\%$ coverage. We also observe, for any of the topologies, that the number of flows routed securely decreases as the number of security levels increases (from left to right in any row). However, even at 4 security levels (common in military networks), a majority of flows was routed securely in the noncongested network. Certainly, congestion dynamics (here, randomly chosen), node labels (here, random), and the different traffic types vary with different networks and will affect the number of flows able to be routed, although this situation can be remedied with conflict minimization. Nonetheless, the results demonstrate that the framework is in fact feasible in several network topologies of different sizes.
\newline \indent
In the case of congested networks, we observe that the linear programming models show slightly better performance in terms of flow assignment compared to the heuristic algorithms. We note that despite the low flow coverage ($\sim50\%$) because new flows could not be supported by the links at some specific time, rules may still have a scheduled install at a delayed time (i.e., when the links can support the new flows), so permitted flows do not necessarily have to be discarded.
\newline \indent
The key insight here is that the randomness used in generating the networks did not hinder the flow coverage. Specifically, in real data-center or cloud networks, edge switches may carry similar traffic~\cite{benson2010network} and thus have similar security levels and only be limited by the available capacity (i.e., not the security levels). These results show that even in the worst case of random level assignment---where for example higher-level nodes may be surrounded by lower-level ones and thus cannot communicate without a policy conflict---paths (even if longer) can be found for a majority of flows.
\newline \indent
Given that a majority of flows can be routed according to policy, we evaluate our model's ability to route the remaining flows (in the noncongested case) to guarantee availability.  Table~\ref{tab:minbenchmarknew} shows that all flows can be routed with minimal policy conflict along the allowed path.} We define policy conflicts as the scenario where a node is visited on a flow path that has a lower security level than the transferred information (i.e., the originating node), and quantify it as the difference of the security levels. {\color{black} Minimizing the conflicts also minimizes the additional security measures needed to protect the flows traversing unsafe links (e.g., via stronger encryption). 
\newline \indent
{\color{black}
For the AS network, paths with no conflict can be found for the majority of flows, while most of the remaining flows only impose a conflict of one security level difference. Less than $5\%$ of flows must be routed through even less secure paths in order to guarantee availability. 
The case is similar for the fat-tree networks; the majority of flow paths have no conflict, approximately $20-30\%$ of flows can be routed with minimal policy conflict of one level, while feasible paths for the remaining flows can also be found, fitting as many flows along two-level difference paths, and so forth.} The key insight here is that most of the remaining flows were able to be routed with a conflict of one security level difference. Investigating the reasons behind this is an interesting direction for future work, where it may be possible to identify whether or not this conflict occurs at \textit{hot} (or commonly used) nodes, and whether that information can be used to relabel nodes (and perhaps repurpose them) or physically reconfigure the network to reduce possible conflicts to a minimum, for any set of flows.
\newline \indent 
We observe that the execution time of computing secure paths is strongly correlated with the number of switches in the network, scaling exponentially. We find that both the LP solver (Gurobi \cite{gurobi}) and the greedy algorithms can quickly compute the paths for a few hundred flows in a \textit{8}-ary tree on average between 0.5s to 1s in a Python-based implementation. When routing thousands of flows for the same network, the execution time can increase to between 7-8s, where the greedy algorithm performs faster (approximately 3.5-4s) as the number of flows increases further. However, large data-centers may be composed of tens of thousands of nodes~\cite{al2008scalable}, and even mid-sized data-centers may contain several hundred or up to one-thousand nodes. In the latter case (for example, represented by a \textit{16}-ary fat-tree), solving the optimization can take nearly 45s. In enterprise-grade networks, the computation speed of MLSNet can be further improved by applying buffering of precomputed paths, implementation on hardware, or clustering of flows which is especially applicable in fat-tree networks.

}
Overall, the results show strong evidence for the ability of our framework to accommodate flows in real networks, while preserving confidentiality of information flow.
\newline \indent

\subsection{Defending Against Attacks}
\label{sec:defenseeffectiveness}

We then implement the network scenario shown in Figure~\ref{fig:evalscenario} with the SDN simulator Mininet~\cite{mininet}, assigning security labels to the nodes and configuring the network to use our framework \textit{MLSNet}. Here, we use SDNMap~\cite{sdnmap} to demonstrate that our framework can mitigate the attack techniques discussed in Section~\ref{sec:securitydiscussion}. SDNMap operates by iteratively probing network nodes with crafted packets and \textit{eavesdropping} on reply messages from all endpoints in the network to reconstruct flow rules (identify active hosts and supported protocols). The gathered information is then used to exploit flow rules and bypass security measures such as access-control lists.

Here, we use the security lattice from Figure~\ref{fig:accesscontrollattice1} to configure the network. Running the MLSNet system at the SDN controller, we assign the security classification of \textit{Public - [ARP,IP,TCP]} to the adversary node at 10.0.0.1 (who is using SDNMap). We then let the adversary begin sending probes into the network. The switch default action for an unknown flow is to send it to the controller for inspection and flow rule generation. On receipt of a new flow, the controller will verify the access control constraint of the communicating parties and the flow control constraint of nodes along a potential flow path. A secure flow rule that obeys the security policy will then be generated.

\textbf{Mitigating lateral movement.}
MLSNet ensures that a node with this classification cannot receive packets from nodes with a higher classification (e.g., from the node at 10.0.0.4 labeled \textit{Secret - [ARP,ICMP,IP,TCP,UDP]}). Therefore, all of the probes destined from the attacker toward a node of higher level should be blocked at the controller, and any induced responses (e.g., from ARP) should also be blocked from flowing back toward the attacker, with no flow rules being generated. We observed exactly this behavior after scanning the network's IP space. The attacker sent out a series of probes enumerating the packet fields (e.g., IP addresses and protocols) to identify active hosts and supported protocols. However, all of the probes sent toward nodes of higher classification were blocked by the controller for not satisfying the access control constraint. Here, SDNMap reported that only the node with IP address 10.0.0.5 replied (the other public host). Although present in the network, the remaining hosts, 10.0.0.4 and 10.0.0.6, are not discovered by the attacker. 

Further, any induced responses over different protocols such as UDP or ICMP were blocked at the controller for not satisfying the flow control constraint as well. Therefore, the attacker could not obtain any other information about the other hosts or protocols supported by them, and their capabilities were only limited to scanning public nodes over the protocols allowed by the defined security label (here, ARP, IP, and TCP). Indeed, future work may investigate optimal labeling and relabeling strategies that can respond to the current network traffic profile in order to dynamically enforce least-privilege and further reduce the threat surface.

\begin{figure}[t]
	\centering
	\includegraphics[width=0.43\textwidth]{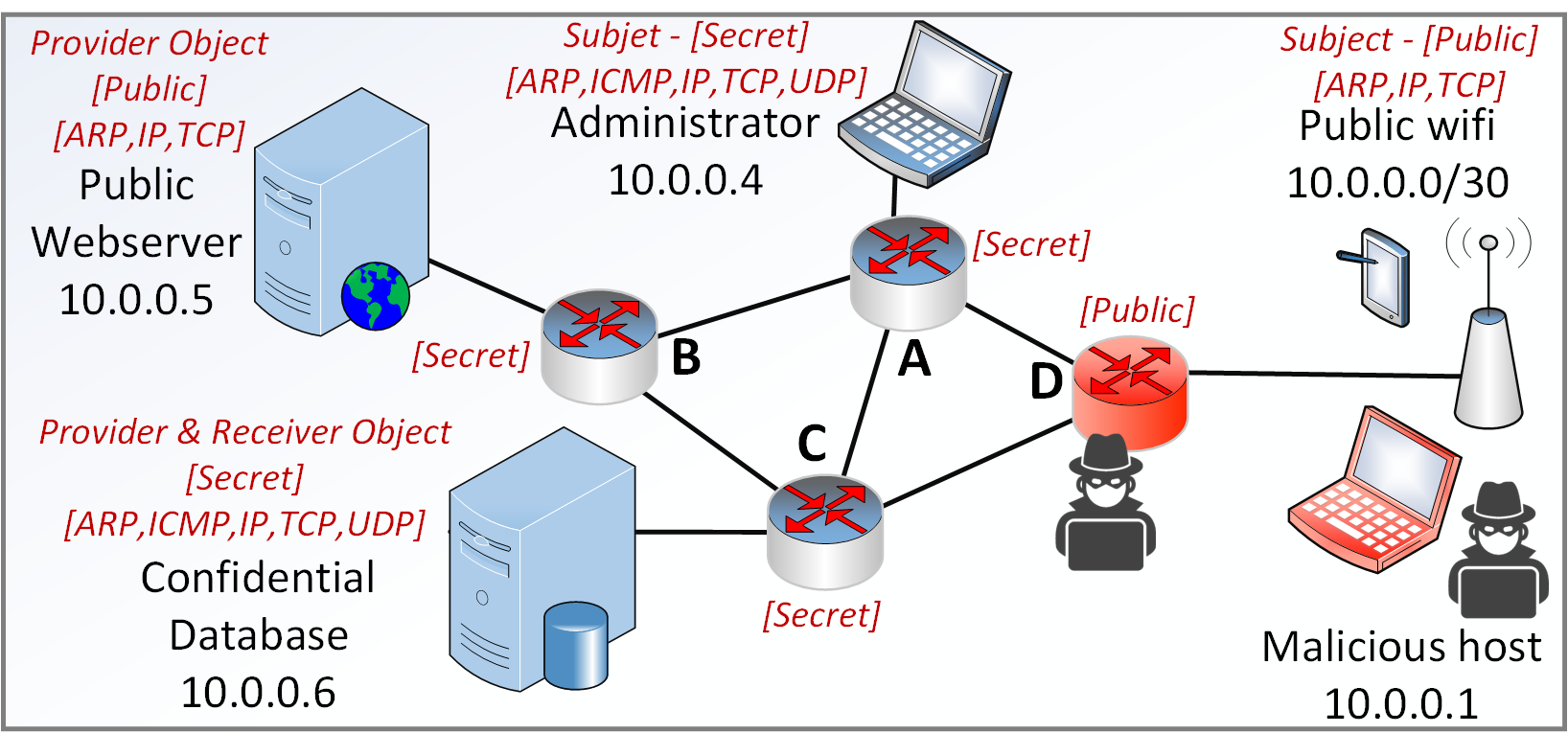}
	\caption{Network scenario with security labels.}
	\label{fig:evalscenario}
\end{figure}

\textbf{Mitigating packet spoofing.}
With respect to packet spoofing, we observe that MLSNet can significantly reduce the threat surface. SDNMap reports that it is not able to spoof IP addresses from hosts which can't be discovered, since our framework will prevent the deployment of a rule for an unrecognized node, and therefore additional flow rule features cannot be reconstructed. By including information from all network layers (as discussed in Section~\ref{sec:secureflowrules}) and deploying rules that only allow flows according to the given policy, we are able to prevent attacks resulting from these forms of spoofing.

Although, we note again that MLS cannot \textit{detect} all forms of packet spoofing, and we rely on other SDN-based defenses to detect packet spoofing of existing nodes or against the topology discovery service~\cite{hong2015poisoning}. Despite this, it can mitigate the threat of data exfiltration from packet spoofing with secure routing paths. For example, consider a compromised node with a lower security level (e.g., public) spoofing a node of higher level (e.g., top-secret) to communicate with another top-secret node, but through two other public nodes. The adversary would have to compromise both public nodes in between in order for them to actually forward the top-secret flow---as it is against security policy.

\textbf{Mitigating eavesdropping.}
In addition to mitigating the threat of lateral movement, MLSNet also prevents the attacker from eavesdropping on communication between nodes of higher-classification. Since the controller prevents flow rules that direct higher-level traffic toward public nodes from being generated, the attacker was unable to capture any traffic passing between the secret and confidential nodes at 10.0.0.4 and 10.0.0.6. Further, any responses induced from the probes were also blocked from generating a new flow rule to direct traffic toward the attacker, preventing eavesdropping. As a result, \textit{MLSNet} is able to significantly limit the ability of SDNMap to reconstruct flow rules.

\section{Related Work}
\label{sec:relatedwork}

Historically, networks have enforced security policies (i.e., information flow) through firewall and routing configuration.  However, these mechanisms are often very coarse and prone to ambiguity, errors, and require coordination across many devices~\cite{reiter-access-control,wool2004quantitative,yuan2006fireman}. Indeed, failures due to errors have enabled a variety of attacks to be launched against real-world networks, including device impersonation, man-in-the-middle, and denial-of-service~\cite{conti2016survey,sommer2010outside,osanaiye2016distributed}. Typically these attacks manifest from a small set of techniques: packet spoofing, lateral movement, and eavesdropping, which have been well-known problems since the 90s~\cite{harris1999tcp} and have become increasingly important as more information is being put online~\cite{shaikh2011security}. In fact, recent work has already demonstrated the ability of an adversary to freely probe within the network to recover sensitive information~\cite{hong2015poisoning,yoon2017flow,dhawan2015sphinx,porras2012security,hu2014flowguard}, including active network hosts and even switch flow table rules in software-defined networks~\cite{achleitner2017adversarial}.

Over time, there have been many defense methods proposed against these techniques, including: source validation to prevent or mitigate packet spoofing~\cite{senie1998network,savu2011cloud}, firewalling to enforce access policies at network boundaries~\cite{yuan2006fireman}, encryption to prevent unauthorized parties from intelligibly interpreting sniffed data, among others. While each useful in a variety of scenarios, they target specific attack techniques and only partially address the problem of confidentiality---ensuring that only authorized entities have access to some data. This motivates our work for developing a solution to combine the benefits offered by each of these methods. We exploit multilevel security to accomplish this, providing guarantees about who in the network may access what data.

Multilevel security allows a network administrator to specify an hierarchical access control policy of a set of subjects on a set of objects. With labels (i.e., a level and categories) given to each subject and object, the policy is enforced through access and flow-control constraints. In fact, multilevel security already plays a critical role in controlling access to information for both military personnel and employees of commercial businesses with different levels of clearance~\cite{saydjari2004multilevel}. Common use cases include controlling file access in an operating system~\cite{loscocco2001security}, object access in generic storage systems~\cite{varadharajan1990multilevel}, table access in a relational database~\cite{qian1997semantic}, as well as a primitive for securing information flow between variables in programming languages~\cite{volpano1997type}.

The notion of multilevel security can also be applied to computer networks, where the MLS policy dictates which nodes are allowed to communicate, what type of traffic they may exchange, and what paths the flows may take in the network. This property precisely address the concerns about confidentiality. We draw inspiration for our framework from the seminal work by \textit{Lu et al.}~\cite{lu1990model} and apply it to SDN-enabled networks. In their work, they introduce a model for multilevel security (MLS) in computer networks by defining a Trusted Network Base (TNB) that is similar to a Trusted Computing Base (TCB) in single-computer systems. The proposed model defines a set of entities (e.g., terminals or printers) and users of the network and relies on the implementation of a security policy by the network endpoints. This approach becomes impractical when having to deploy it on every node in the network, and we exploit the centralization of software-defined networking (SDN)~\cite{shin2012software} to provide this service transparently to the entire network.

\section{Conclusion}
\label{sec:conclusion}
In this paper, we propose \textit{MLSNet}, a framework which can efficiently enforce an MLS policy by generating secure flow-rule configurations. Built upon \textit{access control} and \textit{flow control} constraints, we develop models and heuristic algorithms to compute policy compliant configurations according to two goals: satisfying a strict flow policy and a soft policy. For the deployment of a policy compliant network configuration, we define principles for secure flow rule construction. We then demonstrate that our framework can deploy network configurations able to withstand recently identified attacks on SDNs. We hope this framework will serve as a base for further investigation into defenses which protect the network with a broader scope than specific attacks and efficient mechanisms for resolving policy conflicts in real-time.



%
\newpage
\bibliographystyle{IEEEtran}
\bibliography{main}
\end{document}